\documentclass[letterpaper]{article} 
\usepackage{aaai23}  
\usepackage{times}  
\usepackage{helvet}  
\usepackage{courier}  
\usepackage[hyphens]{url}  
\usepackage{graphicx} 
\def\UrlFont{\rm}  
\usepackage{natbib}  
\usepackage{caption} 
\usepackage[boxed]{algorithm}
\usepackage[noend]{algorithmic}
\usepackage{amsthm,amsmath,amssymb,amsfonts}
\usepackage[utf8]{inputenc}
\usepackage{tikz}
\usepackage{mathtools}
\usepackage{todonotes}
\usepackage{csquotes}
\usepackage{booktabs}
\usepackage{arydshln}
\usepackage{wrapfig}
\usepackage{balance}
\usepackage{curves}
\usepackage{wrapfig}
\usepackage{fancybox}
\usepackage{pdfpages}
\usepackage[export]{adjustbox}
\usepackage[usestackEOL]{stackengine}
\usepackage[flushleft]{threeparttable}
\usepackage{pgfplots}
\usepackage{pgfpages}
\usepackage{pifont}
\usepackage{tablefootnote}
\usepackage{multirow}
\usepackage{comment}

\definecolor{MyGreen}{rgb}{0, 0.7, 0}
\definecolor{MyRed}{rgb}{0.8, 0, 0}
\newcommand{\cmark}{{\color{MyGreen}\ding{51}}}
\newcommand{\xmark}{{\color{MyRed}\ding{55}}}

\newtheorem{definition}{Definition}
\newtheorem{example}{Example}

\newtheorem{theorem}{Theorem}

\usepackage{newfloat}
\usepackage{listings}
\DeclareCaptionStyle{ruled}{labelfont=normalfont,labelsep=colon,strut=off} 
\floatstyle{ruled}
\newfloat{listing}{tb}{lst}{}
\floatname{listing}{Listing}
\title{Daycare Matching in Japan: Transfers and Siblings}

\author{
    Zhaohong Sun\textsuperscript{\rm 1}, Yoshihiro Takenami\textsuperscript{\rm 1}, Daisuke Moriwaki\textsuperscript{\rm 1}, Yoji Tomita\textsuperscript{\rm 1},
    Makoto Yokoo \textsuperscript{\rm 2}
}
\affiliations{
    \textsuperscript{\rm 1}AI Lab, CyberAgent Inc, 
    Japan\\
    \textsuperscript{\rm 2}Kyushu University, Japan\\
}

\begin{document}

\maketitle
\begin{abstract}
In this paper, we study a daycare matching problem in Japan and report the design and implementation of a new centralized algorithm, which is going to be deployed in one municipality in the Tokyo metropolis. There are two features that make this market different from the classical hospital-doctor matching problem: i) some children are initially enrolled and prefer to be transferred to other daycare centers; ii) one family may be associated with two or more children and is allowed to submit preferences over combinations of daycare centers. We revisit some well-studied properties including individual rationality, non-wastefulness, as well as stability, and generalize them to this new setting. We design an algorithm based on integer programming (IP) that captures these properties and conduct experiments on five real-life data sets provided by three municipalities. Experimental results show that i) our algorithm performs at least as well as currently used methods in terms of numbers of matched children and blocking coalition; ii) we can find a stable outcome for all instances, 
although the existence of such an outcome is not guaranteed in theory.
\end{abstract}

\section{Introduction}

Due to a high demand for daycare services and an increasing number of mothers in the workforce, a large number of children are placed on a waiting list each year in Japan, especially in the metropolitan area. The shortage of daycare facilities has become one of the most serious social issues. The Japanese Government provides daycare centers with subsidies to relieve the financial burden of early childhood education. Currently, kindergartens and daycare centers are free of charge for children of age from $3$ to $5$ and fees are deducted for children younger than $3$ years old.
Although the national number of children on the waiting list dropped dramatically to $5634$ in April $2021$ from the peak over $26000$ a few years ago, the government is still concerned about the situation after the COVID-19 pandemic.

The allocation of children to daycare centers in Japan is not based on a first-come, first-served basis. Families with several children first report to their local municipality their preferences over combinations of daycare centers and the municipality then generates a priority ordering over children based on its own scoring system. Children from low-income households and single-parent families, as well as those whose guardians are suffering from diseases or disabilities, usually take precedence over others. Nearly all slots are assigned at the beginning of April each year via a centralized algorithm that takes into consideration of both families' preferences and priorities over children.

The existence of siblings makes this problem resemble \emph{hospital-doctor matching with couples} in which a couple of doctors participates in the job market and submits a joint preference over pairs of hospitals \citep{Roth84a}. Two obvious distinctions are that some children may be initially enrolled and one family may have more than two children. A recent work studies school choice in Chile where both transfers and siblings are considered~\citep{CEE+22a}. However, they assume families have restrictive \emph{higher-first} preferences (i.e. families prioritize the assignment of their children in higher grades), while in the Japanese daycare matching problem families can submit their preferences over any possible combination of daycare centers.

We are actively collaborating with several municipalities and our objective is to design transparent matching algorithms that compute desirable outcomes efficiently. 
\emph{A fundamental question is which properties are deemed suitable and appropriate for this setting.} 
Stability is a standard solution concept for two-sided matching problems and has been widely applied in practice~\citep{Roth08a}. It is often decomposed into individual rationality, non-wastefulness, and fairness in the literature on school choice~\citep{AbSo03b}. We carefully generalize these properties to meet the expectations of municipalities and design a centralized algorithm to achieve them.

We summarize our contributions below:
\begin{itemize}
 \item Our first contribution is to formalize a realistic model of the daycare matching market in Japan based on real-life data sets. To our best knowledge, we are the first to study a two-sided matching problem in which i) children may have 
 their initial enrollment and ii) two or more siblings can submit a joint preference over any possible combination of daycare centers. 
 We compare our model with several representative papers in Table~\ref{table:comparison} and discuss more papers in the section on related work.
 \item We generalize some well-studied concepts including individual rationality, non-wastefulness, and stability by taking the initial enrollment and siblings' joint preferences into account. 
 We adapt these concepts based on the requirements of municipalities; children's welfare is considered more important than daycares'. One subtle difference from previous definitions is that children can make use of other siblings' assignments when forming a blocking coalition.
 \item We develop and implement an algorithm based on integer programming (IP) to capture these new properties, which will be deployed in one city in the Tokyo metropolis\footnote{Our trial matching system is currently being tested and verified by one municipality which is satisfied with our new algorithm.}. 
 We evaluate the performance of our new algorithm through experiments on several real-life data sets. Experimental results show that the outcomes returned by our algorithm are at least as good as the ones yielded by the currently applied methods (some of which are computed by an undisclosed and commercial software package\footnote{\url{https://www.fujitsu.com/global/about/resources/news/press-releases/2018/1112-01.html}}). Another surprising result is that we can find a stable outcome for all instances, although the existence of such an outcome is not guaranteed in theory.
 \end{itemize} 

\begin{table}[tb]
\centering
\resizebox{\columnwidth}{!}{
\begin{tabular}{ccccc}
 \hline
   & transfer & \multicolumn{2}{c}{siblings} & general  
   \\
   &  & $2$ & $\geq 2$ & preferences 
 \\
  \hline
 \citep{BMM14a} & \xmark & \cmark & & \cmark
 \\
 \hline
  \citep{MMT17a} & \xmark & \cmark & & \cmark  
 \\
  \hline
 \citep{HHKS+17a} & \cmark & \xmark & & --
 \\
 \hline
 \citep{Okum19a} & \xmark & \xmark & & --
 \\
 \hline  
 \citep{KaKo19a} & \xmark & \xmark & & --
 \\
 \hline  
  \citep{DMP22a} & \xmark & \cmark & & \xmark   
 \\
 \hline  
  \citep{CEE+22a} & \cmark &  & \cmark & \xmark 
 \\
 \hline  
  Our work & \cmark &  & \cmark & \cmark 
 \\
 \hline  
\end{tabular}
}
\caption{Comparison with some existing models.}
\label{table:comparison}
\end{table}


\section{Related Work}


There exist two existing works on the Japanese daycare matching problem, one of which focused on non-wasteful outcomes and the other one considered fair outcomes. \citet{Okum19a} assumes that schools establish soft quotas for each grade rather than hard bounds. If the soft target for some grade is not fulfilled, then resources (e.g. room space, teachers) assigned to that grade can be utilized by other grades. \citet{Okum19a} proposed a mechanism based on the deferred acceptance that satisfies non-wastefulness, fairness with respect to initial quotas, and fairness for children of the same grade. 
\citet{KaKo19a} studied a matching model under general upper-bounds (i.e., any subset of a feasible set of students remains feasible) that includes the Japanese daycare matching market. They proposed a student-optimal fair matching (SOFM) algorithm which is the most preferred outcome by every student among all feasible, individually rational, and fair matching. 
Although \citep{Okum19a} and \citep{KaKo19a} considered the same research problem, they focus on flexible quotas and do not allow for families' preferences and initial enrollment. 
\citep{DMP22a} consider a school choice model with families where each family has at most two children. They assume that siblings must have the same preferences over daycare centers and siblings are not separable (otherwise they are treated as different families). \citep{CEE+22a} studies school choice in Chile and their setting is the most similar to ours. However, they assume that families have restrictive \emph{higher-first} preferences instead of general preferences over tuples of daycare centers. 

There exists some works on daycare matching which are different from our focus. 
\citep{KMT14a} study a dynamic daycare matching problem and children have individual preferences.
\citep{VBPL17a} study a daycare matching problem in Estonia and they define fairness (or equal access) as the chance for families to be matched to their first choice daycares.
\citep{GJR+21a} propose an integer programming approach for kindergarten allocation in Norway and the objective of their IP models is to minimize the sum of penalties for not assigning children to their prioritized kindergartens. 
\citep{RKG21a} propose and implement a decentralized algorithm
based on deferred acceptance for childcare allocation in two German cities.

There exists a rich literature on hospital-doctor matching with couples which bears similarities to the daycare matching problem. Even though a stable outcome is not guaranteed to exist in theory, \citet{KPR13a} showed that there is a high probability that a stable matching may exist in a large market if the proportion of couples is relatively small. \citet{BMM14a,MMT17a} focused on the computational complexity of computing a stable outcome in the setting of couples (which is NP-hard to decide whether one exists) and proposed algorithms based on Integer Programming and Constraint Programming. \citet{NgVo18a} discussed near-feasible stable matchings for couples under the assumption that hospital capacity is soft. 

\citet{HHKS+17a} considered a school choice problem in which each school is subject to a minimum quota and each student has an initial school and prefers to transfer. They showed that no strategy-proof mechanism exists that is both efficient and fair when minimum quotas exist. They proposed one mechanism called Top Trading Cycles among Representatives with Supplementary Seats (TTCR-SS) which is Pareto efficient and another mechanism called Priority List-based Deferred Acceptance with Minimum Quotas (PLDA-MQ) which is priority list-based stable. \citet{STY18a} also considered the school choice problem with initial schools where distributional constraints can be represented as an M-convex set. They proposed a mechanism based on the Top Trading Cycles mechanism, which satisfies strategy-proofness, feasibility, Pareto efficiency, and individual rationality. However, both works did not consider siblings' joint preferences over schools.

The problem of refugee reallocation also deals with families instead of individuals. In the model of refugee reallocation, each refugee family consists of several members and imposes a multi-dimensional requirement of services including housing, medical services, education, employment, and so on \citep{DKT16a,ACGS18a}. A host locality can accommodate a refugee family only if it meets the multi-dimensional requirement. Each refugee family has a strict preference ordering over host localities and each host locality has a strict preference ordering over refugee families. Thus members from the same refugee family are inseparable and they have an identical preference ordering. However, this is different from our model in which siblings from the same family can be matched in separable daycares and they submit a joint family preference ordering over a tuple of daycares.


\section{Model}

In this section, we introduce the model of daycare matching in Japan based on real-life data sets. 
%
%
%
An instance $I$ consists of a tuple $I = (C, F, D, G, \omega, \succ_C, \succ_F, \succ_D, q)$. 
There are a set of children $C$, a set of families $F$, and a set of daycares $D$. 
Each child $c$ pertains to one family $f \in F$ and each family $f$ is associated with a set of children $C_f \subseteq C$.

The set of children $C$ is partitioned into disjoint families s.t. i) $\bigcup_{f\in F} C_f = C$; ii) for any two different families $f$ and $f'$, $C_{f} \cap C_{f'} = \emptyset$. 
When $C_f$ consists of one child, i.e. $C_f = \{c\}$, we say child $c$ is an \emph{only child} of family $f$. When $C_f$ consists of several children, i.e. $C_f = \{c_1, \ldots, c_k\}$ with $k \geq 2$, we say children $c_1, \ldots, c_k$ are \emph{siblings} of family $f$. 
Each child $c$ belongs to one of the grades $G$ and let $G(c)$ denote the grade of child $c$. Note that siblings from the same family may belong to the same grade, e.g., twins. 

The set of daycares is denoted as $D$. 
Let $d_0 \in D$ denote a dummy daycare representing the option of being unmatched for children.
Some children are initially enrolled at some daycare $d\in D\setminus \{d_0\}$ and prefer to be transferred to a new one. Let $\omega(c) \in D$ denote the initial daycare of child $c$. If $\omega(c) = d_0$, then child $c$ is not initially enrolled (i.e. a new applicant). We generalize the notion of $\omega(c)$ to $\omega(f)$ as follows: given a family $f$ with children $C_f = \{c_1, \ldots, c_k\}$, let $\omega(f) = (\omega(c_1), \ldots, \omega(c_k))$ denote the initial enrollment of each child from $C_f$.
 

There are two types of preference orderings collected in the centralized matching system. Each child $c$ has a strict \emph{individual preference ordering} $\succ_c$ over daycares $D$. We say daycare $d$ is \emph{acceptable} to child $c$ if either it is strictly better than child $c$'s initial enrollment, i.e., $d$ $\succ_c$ $\omega(c)$ or it is the same as child $c$'s initial enrollment, i.e., $d$ $=$ $\omega(c)$. Only acceptable daycares are listed in $\succ_c$ and the dummy daycare $d_0$ is omitted if child $c$ does not have initial enrollment.

Each family $f$ with $k$ children $C_f = \{c_1, \ldots, c_k\}$ where $k \geq 2$ also has a strict \emph{family preference ordering} $\succ_f$ over a tuple of $k$ daycare centers $D^k$, i.e., a joint preference ordering of children $C_f$. Here a tuple of daycares $(d_1, \ldots, d_k) \in D^k$ means that for each integer $i \in \{1, \ldots, k\}$, the $i$-th child $c_i \in C_f$ attends the $i$-th daycare in the tuple. 

During the application phase, if one family's preference type falls into one of predefined categories (e.g., all children must be assigned to the same daycare), then the family just submits its children's individual preference orderings and specifies its preference type. Then the system automatically generates a family preference ordering.
Otherwise, the family needs to fill out its complete preference ordering on its own, which is common in practice. 


In Japan, each municipality has its own scoring system to calculate children's \emph{priority scores}.
Factors depend on families' circumstances, such as parents' employment and health status. Siblings from the same family usually have an identical priority score, although in some cases one sibling may be given additional scores due to some characteristics, e.g., disability. 
In order to derive a strict priority ordering over all children, some small fractions are added for tie-breaking (based on some complex rules). 

%
%

Each daycare also derives a strict priority ordering from children's priority scores. However, daycares may give additional points to some children. For example, if a child has a sibling who is currently enrolled at some daycare, then the child receives extra points for that daycare. 
Note that this is different from the case in other countries where a child is given the highest daycare priority if the child has an enrolled sibling at that daycare \citep{CEE+22a,DMP22a}. 
Another important feature is that if some child $c$ is initially enrolled at daycare $d$, then child $c$ has higher daycare priority at daycare $d$ than those who are not.
Each daycare $d$ has a strict preference ordering $\succ_d$ over $C \cup \{\emptyset\}$, where $\emptyset$ represents the option of leaving some seats vacant. 

The Japanese government imposes two types of feasibility regulations: grade-specific minimum room space for each child and grade-specific maximum teacher-child ratios for each teacher. In the current matching system, each daycare $d$ sets a \emph{hard quota} $q_d^g$ for each grade $g$ that conforms to the two feasibility regulations in advance. Note that the target quotas are designed for new applicants only. If some initially enrolled child transfers to a new daycare, then one additional seat becomes available. 

For the rest of this paper, we just treat each grade $g$ at each daycare $d$ as an independent daycare that only admits children of grade $g$. This simplifies our description of notation and terminology. 
We assume each daycare $d\in D$ has its capacity limit $q_d$. 
We will discuss \emph{flexible quotas} that comply with these two regulations at the end of this paper.


An outcome $\mu$ is a matching between $C$ and $D$ s.t. i) each child $c$ is matched to exactly one daycare $d$
(which can be $d_0$), i.e., $|\mu(c)| = 1$, where $\mu(c)$ denotes the daycare to which $c$ is matched, 
and ii) $\mu(c) = d$ if and only if $c \in \mu(d)$, where $\mu(d)$ denotes the set of children matched to $d$. 
An outcome is feasible if for each daycare $d$, $|\mu(d)| \leq q_d$ holds.
Given outcome $\mu$ and family $f$ with children $C_f= \{c_1, \ldots, c_k\}$, let $\mu(f) = (\mu(c_1), \ldots, \mu(c_k))$ denote the assignment of family $f$ in the outcome $\mu$.

\section{Fundamental Properties}


We next introduce two fundamental properties by generalizing individual rationality and non-wastefulness. Two municipalities confirmed that they consider that any desirable outcome should satisfy these two new properties.
We discuss how to properly define stability in the next section.


\emph{Individual rationality} requires that each child $c$ is matched to some daycare that is weakly better than child $c$'s initial enrollment. Family rationality captures the same idea except that we consider families' preferences instead of children's individual preferences. 
\begin{definition}[Family Rationality]
\label{def:family_rationality}
A feasible outcome $\mu$ satisfies family rationality if for each family $f$, either $\mu(f)$ $\succ_f$ $\omega(f)$ or $\mu(f)$ $=$ $\omega(f)$ holds. 
\end{definition}
%
%


For the classical matching problem without siblings, non-wastefulness requires that 
if a child is not matched to a more preferred daycare, 
then the daycare does not have any vacant seats. 
\citet{CEE+22a} generalize this idea to the setting with multiple siblings; if there exists a pair of daycare $d$ and child $c$ from family $f$ s.t. family $f$ prefers child $c$ to be matched with $d$, 
then daycare $d$ must be full. 
The same idea is also included in the stability concepts for matching models with couples or siblings~\citep{KlKl05a,Klkl07a,McMa10a,DMP22a}.

However, we argue that this generalization may not be adequate. In principle, we believe non-wastefulness should imply that improving the assignment of a family is not possible without hurting other children/families, i.e., the welfare of families must be the first priority. 
We illustrate our concern through Example~\ref{example:issue:NW}.

\begin{example}
\label{example:issue:NW}
Consider one family $f$ with two children $c_1$ and $c_2$ and two daycares $d_1$ and $d_2$ with one seat each. Both children are acceptable to both daycares. Family $f$ have preferences $(d_1, d_2)$ $\succ_f$ $(d_2, d_1)$. 
Outcome $\mu$ $=$ $\{(c_1, d_2)$, $(c_2, d_1)\}$ is considered non-wasteful by \citet{CEE+22a}, but we can obtain another outcome $\mu'$ $=$ $\{(c_1, d_1)$, $(c_2, d_2)\}$ without making any child or family worse off, and it is more preferred by family $f$.
\end{example}

Instead, we propose a new family non-wastefulness concept in Definition~\ref{def:non-wastefulness} that captures the above mentioned idea.

\begin{definition}[Family Non-wastefulness]
\label{def:non-wastefulness}
A feasible outcome $\mu$ satisfies family non-wastefulness, if there does not exist another feasible outcome $\mu'$ and one family $f$ such that i) $\mu'(f) \succ_f \mu(f)$ and ii) for each family $f' \in F \setminus \{f\}$, we have $\mu'(f') = \mu(f')$.
\end{definition}
Definition~\ref{def:non-wastefulness} states that a feasible outcome $\mu$ is family non-wasteful if we cannot find another feasible outcome $\mu'$ and family $f$ such that i) children $C_f$ from family $f$ can be matched to a more preferred tuple of daycares in outcome $\mu'$ and ii) all other families $F \setminus \{f\}$ are matched to the same tuple of daycares in $\mu'$ as in the outcome $\mu$. When each family $f$ has an only child, Definition~\ref{def:non-wastefulness} becomes equivalent to the original concept of non-wastefulness. 


\section{Stability}
In this section, we generalize the original concept of stability to the setting of daycare matching with siblings and initial enrollment. Our new stability concept implies family non-wastefulness in  Definition~\ref{def:non-wastefulness} and it is different from the ones considered in the literature on matching with couples, as our concept allows a child to make use of her siblings' assignment when forming a blocking coalition.

 Recall that in the classical hospital-doctor matching problem, a doctor and a hospital form \emph{a blocking pair} if they are unmatched and prefer to be matched with each other. A feasible outcome is \emph{stable} if it is individually rational and there does not exist a blocking pair~\citep{Roth86a}.
In our setting, a family expresses preferences over a tuple of daycares. Thus we consider a \emph{blocking coalition} between one family and a set of daycares. Briefly speaking, we say a feasible outcome is \emph{stable} if there does not exist any blocking coalition. Note that we treat family rationality as a different property. As explained in the section on Integer Programming, we capture these properties by different constraints separately.


We next introduce a notation $\lambda(\cdot)$ to simplify the definition of blocking coalition and stability. 

\begin{definition}[Function $\lambda$]
Given a feasible outcome $\mu$, a daycare $d$, and two subsets of children $C_1, C_2\subseteq C$, let $\lambda(\mu, d, C_1, C_2)$ denote the number of children who satisfy the following three conditions: 
i) they are matched to daycare $d$ in $\mu$, ii) they are not included in $C_1$ and iii) they have higher daycare priority than at least one child $c \in C_2$.
\end{definition}

We use notation $\lambda(\cdot)$ in the following way: given outcome $\mu$, daycare $d$, family $f$ and a subset of children $C^* \subseteq C_f$ who prefer to be matched to daycare $d$, let $\lambda(\mu, d, C_f, C^*)$ denote the number of children matched to daycare $d$ in $\mu$ excluding $C_f$
 who have higher daycare priority than at least one child $c\in C^*$.

\begin{example}
Consider one daycare $d$ and a set of children $C=\{c_1, c_2, c_3, c_4, c_5\}$ in which children $c_2$ and $c_4$ are siblings while others are only children. Suppose daycare $d$ has three seats and
a priority ordering $\succ_d: c_1, c_2, c_3, c_4, c_5$. 
For matching $\mu = \{(c_1, d), (c_3, d), (c_5, d)\}$, $C_1 = \{c_2, c_4\}$ and $C_2 = \{c_4\}$, we have $\lambda(\mu, d, C_1, C_2) = |\{c_1, c_3\}| = 2$. 
In other words, only two children excluding $C_1$ are matched to daycare $d$ who have higher daycare priority than any child in $C_2$ (i.e., child $c_4$) in the outcome $\mu$.
\end{example}

In order to develop a good intuition, we first confine our attention to the real-life data sets used in our experiments which deal with a setting where each family has at most one pair of twins and no more than three children. We then present a general concept of stability in Definition~\ref{def:stability}. 

We next illustrate four types of blocking coalition for this restrictive setting. Here are three reasons why we break down Definition~\ref{def:stability}. Firstly, it is easier to describe each concept and 
to capture each type of blocking coalition with a corresponding IP constraint. Secondly, if we could not find an outcome that fulfills all constraints, then we can relax them gradually until one exists. Thirdly, the choice of constraints involves a trade-off between the number of matched children and blocking coalitions, i.e., tolerating some form of blocking coalitions may lead to an increase in the number of matched children, as shown in our experimental results.

\begin{definition}[Blocking Coalition I]
\label{def:blocking:only_child}
Given a feasible outcome $\mu$, family $f$ with an only child $c$ and daycare $d$ form a blocking coalition if 
\begin{itemize}
    \item i) $(d) \succ_{f} \mu(f)$;
    \item ii) $\lambda(\mu, d, \{c\}, \{c\}) \leq q_d - 1$.
\end{itemize}
\end{definition}

Definition~\ref{def:blocking:only_child} involves a family with an only child and a daycare, corresponding to the blocking pair in the original stability concept. Condition i) of Definition~\ref{def:blocking:only_child} states that family $f$ prefers daycare $d$ to its assignment $\mu(f)$. Condition ii) of Definition~\ref{def:blocking:only_child} states that child $c$ can be matched to daycare $d$ along with all matched children who have higher priority than child $c$, which contains two cases: either daycare $d$ has at least one vacant seat or child $c$ can replace some matched child $c'$ with lower daycare priority at daycare $d$. 

\begin{definition}[Blocking Coalition II]
\label{def:blocking:twins}
Given a feasible outcome $\mu$, family $f$ with two children $C_f = \{c_1, c_2\}$ of the same age (i.e., twins) and daycare $d$ form a blocking coalition if 
\begin{itemize}
    \item i) $(d, d) \succ_{f} \mu(f)$;
    \item ii) $\lambda(\mu, d, C_f, C_f) \leq q_d - 2$.
\end{itemize}
\end{definition}

Definition~\ref{def:blocking:twins} involves a family $f$ with a pair of twins $c_1$ and $c_2$ who apply to the same daycare $d$. Condition ii) of Definition~\ref{def:blocking:twins} states that daycare $d$ can admit children $C_f$ as well as all matched children who have higher priority than 
either child $c_1$ or $c_2$.


\begin{definition}[Blocking Coalition III] 
\label{def:blocking:distinct}
Given a feasible outcome $\mu$, family $f$ with children $C_f = \{c_1, \ldots, c_k\}$, and a tuple of daycares $(d_1, \ldots, d_k)$ with $k$ distinct daycares, family $f$ and
daycares $d_1, \ldots, d_k$ form a blocking coalition if 
\begin{itemize}
    \item i) $(d_1, \ldots, d_k)$ $\succ_{f}$ $\mu(f)$;
    \item ii) $\forall d_i \in (d_1, \ldots, d_k)$, $\lambda(\mu, d_i, C_f, \{c_i\}) \leq q_{d_i} - 1$.
\end{itemize}
\end{definition}

Definition~\ref{def:blocking:distinct} involves a family $f$ with $k$ children $\{c_1$, $\ldots$, $c_k\}$ who apply to $k$ distinct daycares $(d_1$, $\ldots$, $d_k)$. Condition ii) of Definition~\ref{def:blocking:distinct} states that for each daycare $d_i$ from the tuple of daycares $(d_1$, $\ldots$, $d_k)$, daycare $d_i$ can admit the corresponding child $c_i$ from $C_f$ and all matched children (excluding other siblings from $C_f$) who have higher priority than child $c_i$.

\begin{definition}[Blocking Coalition IV]
\label{def:blocking:mixed}
Given a feasible outcome $\mu$, family $f$ with children $C_f = \{c_1, c_2, c_3\}$ and a tuple of daycares $(d_1, d_1, d_2)$ with $d_1 \neq d_2$\footnote{Here we assume if a family has three children including one pair of twins, 
$c_1$ and $c_2$ are twins and $c_3$ is another child.}
, family $f$ and daycares $d_1$ and $d_2$ form a blocking coalition if 
\begin{itemize}
    \item i) $(d_1, d_1, d_2)$ $\succ_{f}$ $\mu(f)$;
    \item ii) For daycare $d_1$, $\lambda(\mu, d_1, C_f, \{c_1, c_2\}) \leq q_{d_1} - 2$ and for daycare $d_2$, $\lambda$$(\mu$, $d_2$, $C_f$, $\{c_3\})$ $\leq$ $q_{d_2}$ $-$ $1$.
\end{itemize}
\end{definition}
Definition~\ref{def:blocking:mixed} involves family $f$ with a pair of twins applying to the same daycare $d_1$ and another sibling applying to a different daycare $d_2$. Condition ii) of Definition~\ref{def:blocking:mixed} states that daycare $d_1$ can admit children $c_1$ and $c_2$ as well as all matched children (excluding $C_f$) who have higher priority than either $c_1$ or $c_2$, and daycare $d_2$ can admit child $c_3$ along with all matched children (excluding $C_f$) who have higher priority than child $c_3$.



We next present the general concept of stability without imposing any limitation on the number of siblings or twins. Prior to that, we first introduce a new concept called demand table in Definition~\ref{def:demand_table}.

\begin{definition}[Demand Table]
\label{def:demand_table}
Given a family $f$ and a tuple of daycares $(d_1, \ldots, d_k)$ in family preference $\succ_f$,  \emph{demand table} $T$ $=$ $\{(d_1: C_1)$, $\ldots$, $(d_{k'}: C_{k'})\}$ is given as follows: for each distinct daycare $d$ in $(d_1, \ldots, d_k)$, add an entry in the form of $(d: \{c'_{1}, c'_{2}, \ldots\})$ where $c'_{1}, c'_{2}, \ldots$ are a subset of children from $C_f$ who apply to daycare $d$ w.r.t. the tuple $(d_1, \ldots, d_k)$. 
\end{definition}

\begin{example}[Instance of Demand Table]
\label{example:damand_table}
Consider one family $f$ with children $C_f = \{c_1, c_2, c_3\}$ and a tuple of daycares $(d_1, d_1, d_2)$ in the family preference ordering $\succ_f$. The corresponding demand table $T$ is $\{(d_1: \{c_1, c_2\}), (d_2: \{c_3\})\}$.
\end{example}

\begin{definition}[Stability]
\label{def:stability}
Given a feasible outcome $\mu$, family $f$ with children $C_f = \{c_1, \ldots, c_k\}$ and a tuple of daycares $(d_1, \ldots, d_k)$, family $f$ and all distinct daycares in the tuple $(d_1, \ldots, d_k)$ form a blocking coalition if 
\begin{itemize}
    \item i) $(d_1, \ldots, d_k)$ $\succ_{f}$ $\mu(f)$ and
    \item ii) let $T = \{(d_1: C_1), \ldots, (d_{k'}: C_{k'})\}$ denote the demand table w.r.t. $(d_1, \ldots, d_k)$, then for each pair of daycare $d$ and $C_d$ in $T$, we have 
    \[\lambda(\mu, d, C_f, C_d) \leq q_{d} - |C_d|.\]
\end{itemize}
A feasible outcome $\mu$ is stable if there does not exist any blocking coalition.
\end{definition}


We next prove that there may not exist any stable outcome in general. However, in our experiments using real-life data sets, we found that a stable outcome exists for all instances.

\begin{theorem}
\label{theo:non-existence:stable}
The set of stable outcomes may be empty even if there are only four children and two daycares.
\end{theorem}

\begin{proof}
We prove Theorem~\ref{theo:non-existence:stable} through the following counter-example. Consider three families $f_1$ with two children $C_{f_1} = \{c_1, c_2\}$, $f_2$ with one child $C_{f_2} = \{c_3\}$, $f_3$ with one child $C_{f_3} = \{c_4\}$, and two daycares $d_1$ with $2$ seats and $d_2$ with $1$ seat. The preference and priority profiles are as follows:
\vspace{-1mm}
 \begin{align*}
 & \succ_{f_1}: (d_1, d_1), \quad \succ_{f_2}: (d_2), (d_1) \quad \succ_{f_3}: (d_1), (d_2) 
\\
 & \succ_{d_1}: c_1, c_3, c_2, c_4 \qquad \succ_{d_2}: c_4, c_3
\end{align*}

There are five feasible and family rational outcomes to which we cannot add any more children without violating feasibility. We next show that none of them satisfies stability.
\begin{itemize}
    \item For outcome $\mu_1 = \{(c_1, d_1), (c_2, d_1), (c_3, d_2)\}$, family $f_3$ can form a blocking coalition with daycare $d_2$, as 
    $\lambda(\mu_1, d_2, \{c_4\}, \{c_4\}) = 0 \leq q_{d_2} - 1$. In other words, no child matched to daycare $d_2$ in $\mu_1$ has a higher daycare priority than child $c_4$.
    \item For outcome $\mu_2 = \{(c_1, d_1), (c_2, d_1), (c_4, d_2)\}$, family $f_2$ can form a blocking coalition with daycare $d_1$, as 
    $\lambda(\mu_2, d_1, \{c_3\}, \{c_3\}) = 1 \leq q_{d_1} - 1$. In other words, only one child, i.e., $c_1$, who is matched to daycare $d_1$ in $\mu_2$, has a higher daycare priority than child $c_3$. 
    \item For outcome $\mu_3 = \{(c_4, d_1), (c_3, d_2)\}$, family $f_1$ can form a blocking coalition with daycare $d_1$, as 
    $\lambda(\mu_3, d_1, C_{f_1}, C_{f_1}) = 0 \leq q_{d_1} - 2$. In other words, no child matched to daycare $d_1$ in $\mu_3$ has a higher daycare priority than either child $c_1$ or $c_2$. 
    \item For outcome $\mu_4 = \{(c_4, d_1), (c_3, d_1)\}$, family $f_2$ can form a blocking coalition with daycare $d_2$, as 
    $\lambda(\mu_4, d_2, \{c_3\}, \{c_3\}) = 0 \leq q_{d_2} - 1$. In other words, no child matched to daycare $d_2$ in $\mu_4$ has a higher daycare priority than child $c_3$. 
    \item For outcome $\mu_5 = \{(c_3, d_1), (c_4, d_2)\}$, family $f_3$ can form a blocking coalition with daycare $d_1$, as 
    $\lambda(\mu_5, d_1, \{c_4\}, \{c_4\}) = 1 \leq q_{d_1} - 1$. In other words, only one child, i.e., $c_3$, who is matched to daycare $d_1$ in $\mu_5$, has a higher daycare priority than child $c_4$.
\end{itemize}
This completes the proof of Theorem~\ref{theo:non-existence:stable}.
\end{proof}


\section{Integer Programming}

In this section, we present a practical algorithm based on Integer Programming (IP) for the real-life data sets provided by three municipalities where each family has no more than three children and at most one pair of twins. 

We first introduce another representation of preference orderings derived from family preferences used in the IP algorithm. Given family $f$ with $k$ children $C_f = \{c_1, \ldots, c_k\}$ and family preference ordering $\succ_f$, a \emph{projected preference ordering} $\succ_{f,c_i}$ of child $c_i \in C_f$ consists of the $i$-th daycare in each tuple of daycares in $\succ_f$. 
For an only child $c$, there is no difference among $\succ_c$, $\succ_f$ and $\succ_{f, c}$. We next illustrate the relationship of $\succ_c, \succ_f, \succ_{f, c}$ for child $c$ who has a sibling through Example~\ref{example:preferences_relationship}. 
\begin{example}
\label{example:preferences_relationship}
Consider one family $f$ with two children $C_{f} = \{c_1, c_2\}$ and three daycares $D=\{d_0, d_1, d_2\}$ (including a dummy daycare $d_0$). 
Suppose both children $c_1$ and $c_2$ have the same individual preference ordering $d_1 \succ_c d_2$. 

Family $f$ requires that either both children are enrolled at the same daycare or if only one child can be matched, then higher precedence is given to child $c_1$. This rules out the possibility that two children are enrolled at different daycares.
The family preference ordering $\succ_f$ and the projected preference ordering $\succ_{f,c}$ of each child $c$ are as follows:
\begin{align*}
& \succ_f: (d_1, d_1), (d_2, d_2), (d_1, d_0), (d_2, d_0), (d_0, d_1), (d_0, d_2)
\\
& \succ_{f, c_1}: d_1, d_2, d_1, d_2, d_0, d_0
\\
& \succ_{f, c_2}: d_1, d_2, d_0, d_0, d_1, d_2
\end{align*}
\end{example}
As shown in Example~\ref{example:preferences_relationship}, a given daycare $d$ may appear multiple times in the projected preference ordering $\succ_{f, c}$ and the dummy daycare $d_0$ cannot be omitted in $\succ_{f,c}$ if child $c$ has another sibling. 

Given a projected preference ordering $\succ_{f,c}$ of child $c$ and a position $p \in [1, \mid\succ_{f,c}\mid]$ in the range from $1$ to the length of child $c$'s projected preference ordering $\succ_{f,c}$, let $d(c, p)$ denote the daycare at position $p$ in $\succ_{f,c}$ and let $P(c, d) = \{p:d(c, p) = d\}$ denote the set of positions corresponding to daycare $d$ in $\succ_{f,c}$.  

We next introduce four types of binary variables used in the IP algorithm.
For each child $c \in C$ and each position $p \in [1, \mid \succ_{f,c} \mid]$, create a variable $y_{c, p}$ indicating whether child $c$ is matched to the $p$-th element in $\succ_{f,c}$~\footnote{It is inaccurate to state that ``child $c$ is matched to the $p$-th `daycare' in $\succ_{f,c}$'', as the same daycare at position $p$ may appear multiple times in $\succ_{f,c}$.}.
\begin{equation}
\label{var:y}
y_{c, p} = 
\begin{cases}
1 \quad\text{\small{if $c$ is matched to the $p$-th element in $\succ_{f,c}$}}
\\
0 \quad\text{\small{otherwise.}}
\end{cases}
\end{equation}

For each child $c \in C$ and each position $p \in [1, \mid \succ_{f,c} \mid]$, create a variable $\alpha_{c, p}$ indicating whether child $c$ is matched to the $r$-th element in $\succ_{f,c}$ with $r \leq p$ as follows:
\begin{equation}
\label{var:alpha}
\alpha_{c, p} = 
\sum_{r=1}^p y_{c, r}.
\end{equation}
%
For each child $c \in C$ who has a sibling and each position $p \in [1, \mid \succ_{f,c} \mid]$, create a variable $\beta_{c, p}$ indicating whether child $c$ can be matched to the $p$-th element without modifying other families' assignment. Please note that $\beta_{c, p} = 1$ if child $c$ cannot be matched to the $p$-th element without changing other families' assignment, and $\beta_{c, p} = 0$ otherwise.
There are two cases when determining the value of $\beta_{c, p}$. The first case is that child $c$ from family $f$ has a twin $c_{t}$ and both children apply to the same daycare at position $p$, i.e., $ d(c, p) = d(c_{t},p)$.
\begin{equation}
\label{var:beta:1}
\beta_{c, p} = 
\begin{cases}
1 \quad \text{\small{if $d(c, p)$ admits at least}}
\\  
\ \ \quad \text{\small{$q_d-1$ children except for $C_f$}}
\\
0 \quad \text{\small{otherwise.}}
\end{cases}
\end{equation}
The second case is that child $c$ and all of her siblings apply to different daycares at position $p$:
\begin{equation}
\label{var:beta:2}
\beta_{c, p} = 
\begin{cases}
1 \quad \text{\small{if $d(c, p)$ admits $q_d$ children except for $C_f$}}
\\
0 \quad \text{\small{otherwise.}}
\end{cases}
\end{equation}
%
For each child $c \in C$ who has a sibling and each position $p \in [1, \mid \succ_{f,c} \mid]$, create a variable $\gamma_{c, p}$ indicating whether child $c$ can coexist with all matched children who have higher priority than either $c$ or his twin $c_t$ (if any).
Similar to variables $\beta$, variable $\gamma_{c, p} = 1$ means that child $c$ cannot coexist with all matched children with higher priority than either $c$ or her twin $c_t$ (if any), and $\gamma_{c, p} = 0$ otherwise. 
We also consider two cases depending on whether child $c$ has a twin $c_t$ who applies to the same daycare at position $p$. 
The first case is that child $c$ from family $f$ has a twin $c_{t}$ who applies to the same daycare at position $p$, i.e. $d(c, p) = d(c_{t}, p)$:
\begin{equation}
\label{var:gamma:1}
\gamma_{c, p} = 
\begin{cases}
1 \ \ \text{\small{if $d(c, p)$ admits at least $q_d-1$ children}}
\\  
\ \ \ \ \text{\small{with higher priority than $c$ or $c_t$} except for $C_f$ }
\\
0 \ \ \text{\small{otherwise.}}
\end{cases}
\end{equation}
The second case is that child $c$ and all of her siblings apply to different daycares: 
\begin{equation}
\label{var:gamma:2}
\gamma_{c, p} = 
\begin{cases}
1 \ \ \text{\small{if $d(c, p)$ admits exactly $q_d$ children}}
\\  
\ \ \ \ \text{\small{with higher priority than $c$} except for $C_f$}
\\
0 \ \ \text{\small{otherwise.}}
\end{cases}
\end{equation}

Recall that an outcome is feasible if i) each child $c\in C$ is matched to at most one daycare; and ii) each daycare $d\in D$ can accommodate at most $q_d$ children. We capture these two requirements through Constraints~\ref{LP:feasibility1} and \ref{LP:feasibility2}, respectively.  
\begin{equation}
\label{LP:feasibility1}
    \sum_{p=1}^{\mid \succ_{f,c} \mid} y_{c, p} \leq 1 \qquad  \forall c\in C
\end{equation}
\begin{equation}
\label{LP:feasibility2}
    \sum_{c\in C} \sum_{p\in P(c, d)} y_{c, p} \leq q_d \qquad  \forall d\in D
\end{equation}
As we use children's projected preferences, we need to ensure that siblings must be matched to a tuple of daycares that are located at the same position of each child's projected preference ordering. That is, for each family $f$ with children $C_f = \{c_1, ... ,c_k\}$, for each position $p \in [1, \mid\succ_{f,c}\mid]$, 
\begin{equation}
\label{LP:feasibility3}
    y_{c_1, p} = ... = y_{c_k, p}. 
\end{equation}

Family rationality is guaranteed by Constraint~\ref{LP:FR} combined with Constraints~\ref{LP:feasibility3}.
For each child $c\in C$ with $\omega(c) = d \neq d_0$, we have
\begin{equation}
\label{LP:FR}
    \sum_{p=1}^{\mid \succ_{f,c} \mid} y_{c, p} = 1.
\end{equation}



\subsection{Non-wastefulness}

We capture non-wastefulness with the following four constraints.

Constraint~\ref{LP:NW-I} corresponds to the blocking coalition concept in Definition~\ref{def:blocking:only_child} which involves families with an only child. 
For each family $f$ with an only child $c$ and each position $p \in [1, |\succ_{f, c}|]$,
let $d$ denote the daycare at position $p$,  
\begin{equation}
\label{LP:NW-I}
    (1 - \alpha_{c,p}) * q_d \leq\sum_{c'\in C} \sum_{p\in P(c', d)} y_{c', p} 
\end{equation}

Constraint~\ref{LP:NW-II} corresponds to the blocking coalition concept in Definition~\ref{def:blocking:twins} which involves families with a pair of twins who apply to the same daycare. 
For each family $f$ with twins denoted by $c_1$ and $c_2$, as well as each position $p$ such that twins $c_1$ and $c_2$ apply to the same daycare at position $p$, let $c^*$ denote either $c_1$ or $c_2$, then we have

\begin{equation}
\label{LP:NW-II}
    (1 - \alpha_{c^*,p}) \leq \beta_{c^*, p}.
\end{equation}

Constraint~\ref{LP:NW-III} corresponds to the blocking coalition concept in Definition~\ref{def:blocking:distinct} which involves families with siblings applying to all different daycares. 
For each family $f$ with children $C_f$ and each position $p$ such that siblings $C_f$ apply to all different daycares at position $p$, the following holds for any child $c \in C_f$, 

\begin{equation}
\label{LP:NW-III}
    (1 - \alpha_{c,p}) \leq \sum_{c' \in C_f}\beta_{c', p}.
\end{equation}

Constraint~\ref{LP:NW-IV} corresponds to the blocking coalition concept in Definition~\ref{def:blocking:mixed} which involves families with a pair of twins applying to the same daycare and a third sibling applying to a different daycare. For each family $f$ with children $C_f = \{c_1, c_2, c_3\}$ and a tuple of daycares $(d_1, d_1, d_2)$ with $d_1 \neq d_2$ at position $p$ in $\succ_f$, let $c^*$ denote either $c_1$ or $c_2$, then the following holds, 
\begin{equation}
\label{LP:NW-IV}
    (1 - \alpha_{c^*,p}) \leq \beta_{c^*, p} + \beta_{c_3, p}.
\end{equation}

\subsection{Stability}
We next explain how to capture stability with the following four constraints corresponding to four types of blocking coalitions.
%
%
Constraint~\ref{LP:stable-I} corresponds to the blocking coalition concept in Definition~\ref{def:blocking:only_child} which involves families with an only child. 
For each family $f$ with an only child $c$ and each position $p \in [1, \mid \succ_{f,c} \mid]$,
let $d$ denote the daycare at position $p$ and let $C^* \subseteq C$ denote a set of children who have higher priority than child $c$ at daycare $d$, 
\begin{equation}
\label{LP:stable-I}
    (1 - \alpha_{c,p}) * q_d \leq\sum_{c'\in C^*} \sum_{p'\in P(c', d)} y_{c', p'}. 
\end{equation}


\begin{proof}
Consider any family $f$ with an only child $c$ and any position $p \in [1, \mid \succ_{f,c} \mid]$. 

If child $c$ is matched to the $r$-th element in her projected preference ordering $\succ_{f,c}$ with $r\leq p$, i.e., $\alpha_{c,p} = 1$, then Constraint~\ref{LP:stable-I} is trivially satisfied and 
child $c$ will not form a blocking coalition with the daycare at position $p$ as child $c$ is matched to a weakly better element at position $r$ (violating condition i) of Definition~\ref{def:blocking:only_child}). 

If child $c$ is not matched to the $r$-th element in $\succ_{f,c}$ with $r\leq p$, i.e., $\alpha_{c,p} = 0$, then Constraint~\ref{LP:stable-I} requires that the number of matched children with higher priorities than child $c$ at daycare $d$ reaches daycare $d$'s capacity $q_d$. In that case, daycare $d$ will not form a blocking coalition with child $c$, as daycare $d$ is full and all matched children have higher daycare priorities than child $c$ (violating condition ii) of Definition~\ref{def:blocking:only_child}).

Thus family $f$ and daycare $d$ at position $p$ will not form a blocking coalition.
\end{proof}


Constraint~\ref{LP:stable-II} corresponds to the blocking coalition concept in Definition~\ref{def:blocking:twins} which involves families with a pair of twins who apply to the same daycare. 
For each family $f$ with twins denoted by $c_1$ and $c_2$, as well as each position $p$ such that twins $c_1$ and $c_2$ apply to the same daycare at position $p$, let $c^*$ denote the child who has lower priority at daycare $d(c^*,p)$ between $c_1$ and $c_2$, then we have 
\begin{equation}
\label{LP:stable-II}
    (1 - \alpha_{c^*,p}) \leq \gamma_{c^*, p}.
\end{equation}


\begin{proof}
Consider any family $f$ with a pair of twins denoted by $c_1$ and $c_2$ as well as any position $p$ such that children $c_1$ and $c_2$ apply to the same daycare at position $p$. Let $c^*$ denote the child who has a lower priority at daycare $d(c^*, p)$ between $c_1$ and $c_2$.

If child $c^*$ is matched to the $r$-th element in her projected preference ordering $\succ_{f,c^*}$ with $r\leq p$, i.e., $\alpha_{c^*,p} = 1$, then Constraint~\ref{LP:stable-II} is trivially satisfied as $\gamma_{c^*, p}$ could be either $0$ or $1$. 
Family $f$ will not form a blocking coalition with the daycare at position $p$ as family $f$ is matched to a weakly better tuple of daycares at position $r$ (violating condition i) of Definition~\ref{def:blocking:twins}).

If child $c^*$ is not matched to the $r$-th element in $\succ_{f,c^*}$ with $r\leq p$, i.e., $\alpha_{c^*,p} = 0$, then Constraint~\ref{LP:stable-II} requires that $\gamma_{c^*, p} = 1$. By
Constraint~\ref{var:gamma:1}, $\gamma_{c^*, p} = 1$ implies that daycare $d(c, p)$ admits at least $q_d-1$ children excluding $C_f$ who have a higher priority than child $c^*$. Thus daycare $d$ will not form a blocking coalition with family $f$, as children $C_f$ cannot coexist with all matched children who have a higher daycare priority than child $c^*$ (violating condition ii) of Definition~\ref{def:blocking:twins}).

Thus family $f$ and daycare $d$ at position $p$ will not form a blocking coalition.
\end{proof}


Constraint~\ref{LP:stable-III} corresponds to the blocking coalition concept in Definition~\ref{def:blocking:distinct} which involves families with siblings applying to all different daycares. 
For each family $f$ with children $C_f$ and each position $p$ such that siblings $C_f$ apply to all different daycares at position $p$, the following holds for any child $c \in C_f$, 
\begin{equation}
\label{LP:stable-III}
    (1 - \alpha_{c,p})  \leq \sum_{c' \in C_f}\gamma_{c', p}.
\end{equation}


\begin{proof}
Consider any family $f$ with $k (\leq 3)$ children $C_f = \{c_1, \ldots, c_k\}$ and any position $p$ such that children 
$C_f$ apply to $k$ distinct daycares $(d_1, \ldots, d_k)$ at $p$ in $\succ_f$. Let $c$ denote any child from family $f$, i.e., $c \in C_f$.

If child $c$ is matched to the $r$-th element in $\succ_{f,c}$ with $r\leq p$, i.e., $\alpha_{c,p} = 1$, then Constraint~\ref{LP:stable-III} is trivially satisfied as $\gamma_{c', p}$ could be either $0$ or $1$ for any child $c' \in C_f$. 
In this case, family $f$ will not form a blocking coalition with daycares $(d_1, \ldots, d_k)$ at position $p$, as family $f$ is matched to a weakly better tuple of daycares at position $r$ (violating condition i) of Definition~\ref{def:blocking:twins}).

If child $c$ is not matched to the $r$-th element in $\succ_{f,c}$ with $r\leq p$, i.e., $\alpha_{c,p} = 0$, then Constraint~\ref{LP:stable-III} requires that $\sum_{c' \in C_f}\gamma_{c', p} \geq 1$. In other words, there exists at least one child $c' \in C_f$ with $\gamma_{c', p} = 1$. By Constraint~\ref{var:gamma:2}, $\gamma_{c', p} = 1$ implies that daycare $d = d(c',p)$ admits at least $q_d$ children excluding $C_f$ who have a higher priority than child $c'$. Thus daycare $d(c',p)$ will not form a blocking coalition with family $f$, as child $c'$ cannot coexist with all matched children who have higher daycare priorities than child $c'$ (violating condition ii) of Definition~\ref{def:blocking:twins}).

Thus family $f$ and a tuple of daycares $(d_1, \ldots, d_k)$ at position $p$ in $\succ_f$ will not form a blocking coalition.
\end{proof}


Constraint~\ref{LP:stable-IV} corresponds to the blocking coalition concept in Definition~\ref{def:blocking:mixed} which involves families with a pair of twins applying to the same daycare and a third sibling applying to a different daycare. For each family $f$ with children $C_f = \{c_1, c_2, c_3\}$ and a tuple of daycares $(d_1, d_1, d_2)$ with $d_1 \neq d_2$ at position $p$ in  $\succ_f$, let $c^*$ denote the child who has lower priority at daycare $d_1$ between $c_1$ and $c_2$, then the following holds, 
\begin{equation}
\label{LP:stable-IV}
    (1 - \alpha_{c^*,p}) \leq \gamma_{c^*, p} + \gamma_{c_3, p}.
\end{equation}

\begin{proof}
Consider any family $f$ with three children $C_f = \{c_1, c_2, c_3\}$ who apply to a tuple of daycares $(d_1, d_1, d_2)$ with $d_1 \neq d_2$ at any position $p$ in $\succ_f$. Let $c^*$ denote the child who has a lower priority at daycare $d_1$ between $c_1$ and $c_2$.

If child $c^*$ is matched to the $r$-th element in her projected preference ordering $\succ_{f,c^*}$ with $r\leq p$, i.e., $\alpha_{c^*,p} = 1$, then Constraint~\ref{LP:stable-IV} is trivially satisfied as $\gamma_{c^*, p}$ and $\gamma_{c_3, p}$ could be either $0$ or $1$. 
Family $f$ will not form a blocking coalition with the daycare at position $p$ as family $f$ is matched to a weakly better tuple of daycares at position $r$ (violating condition i) of Definition~\ref{def:blocking:twins}).

If child $c^*$ is not matched to the $r$-th element in $\succ_{f,c^*}$ with $r\leq p$, i.e., $\alpha_{c^*,p} = 0$, then Constraint~\ref{LP:stable-IV} requires that $\gamma_{c^*, p} + \gamma_{c_3, p} \geq 1$. Then either $\gamma_{c^*, p} = 1$ or $\gamma_{c_3, p} = 1$ holds.
By Constraint~\ref{var:gamma:1}, $\gamma_{c^*, p} = 1$ implies that daycare $d_1$ admits at least $q_{d_1}-1$ children excluding $C_f$ who have a higher priority than child $c^*$. By Constraint~\ref{var:gamma:2}, $\gamma_{c_3, p} = 1$ implies that daycare $d_2$ admits at least $q_{d_2}$ children excluding $C_f$ who have higher priorities than child $c_3$.
For both cases, either daycare $d_1$ or daycare $d_2$ will not form a blocking coalition with family $f$, as children $C_f$ cannot coexist with all matched children who have higher daycare priorities than child $c^*$ or child $c_3$ (violating condition ii) of Definition~\ref{def:blocking:twins}).

Thus family $f$ and daycares $d_1$ and $d_2$ will not form a blocking coalition.
\end{proof}


The objective of the IP algorithm is to maximize the total number of matched children as described in Constraint~\ref{LP:objective}. Recall that a dummy daycare $d_0$ can be included in children's projected preferences, so we must remove the number of children who are matched to daycare $d_0$. 
\begin{equation}
\label{LP:objective}
 \max   \sum_{c \in C} \sum_{p=1}^{\mid\succ_{f,c}\mid} y_{c, p} - \sum_{c \in C} \sum_{p' \in P(c, d_0)} y_{c, p'}
\end{equation}

\section{Experiments}

In this section, we evaluate the performance of our new algorithm by running experiments on five real-life data sets provided by the following three municipalities.
%
Shibuya ward is one of the major commercial centers in Tokyo with a population of more than $220,000$.
Tama is a rural city located in the west of the Tokyo Metropolis with a population of less than $150,000$.
Moriguchi is one of the satellite cities of the Osaka Metropolis with a population of around $140,000$.

All experiments were conducted on a laptop with an Apple M1-Max processor and 32Gb of memory. Our IP model was implemented via Google OR-Tools using the default CP-SAT solver~\footnote{https://developers.google.com/optimization}. For all five data sets, stable outcomes were computed in no more than 2 seconds. We also run more experiments, where we restrict IP constraints; the running times vary from a few seconds to around 40 minutes. The details of these experiments are presented in Appendix.


We summarize the basic information of the five data sets in Table~\ref{tab:data}. We found three common features in these data sets. Firstly, each family contains no more than three children and at most one pair of twins. 
Secondly, around $15\%$ of children have siblings and more than $7\%$ of children prefer to transfer. 
Thirdly, daycares are suffering from a shortage of seats for ages $0$ and $1$, but there is an excess of seats for ages $3$ and above. We give more details about the imbalance of demand and supply by age in  Appendix.

For Shibuya ward, the status quo outcome is calculated by the commercial software package mentioned in Introduction. We have no access to the detailed implementation of that algorithm and it is not clear whether it satisfies the desirable theoretical properties discussed in this paper.
%
For Tama city and Moriguchi city, the status quo outcome is calculated manually and it takes several government officials one or two weeks to determine and verify the outcome. 
%

We compare the outcomes yielded by IP and the status quo in Table~\ref{table:outcome} and we can draw two conclusions from these experiments: i) our IP algorithm outperforms the currently applied methods used by Tama and Moriguchi city and performs the same as the commercial algorithm for Shibuya ward~\footnote{Although our IP algorithm returned exactly the same outcome as the status quo (i.e., the commercial software) for both Shibuya data sets, the results can be different for other cases. For instance, Tama also purchased the same software, but the performance was not as good as the status quo (i.e., the manual method) and they decided not to deploy it and chose our new algorithm instead.}, ii) there always exists a stable outcome for five real-life data sets. Our new algorithm is accepted by Tama city and will be deployed in its matching system shortly.

Due to the confidentiality agreement with local municipalities, we do not have the authorization to disclose these data sets and our codes, but we have given full description of our IP algorithm for reproducibility purpose.

\begin{table}[tb]
\begin{center}
\resizebox{\columnwidth}{!}{
\begin{tabular}{cccccccc}
    \hline
    \multirow{2}{*}{}
    & \# children in & \multicolumn{2}{c}{Tama} & \multicolumn{2}{c}{Shibuya} & Moriguchi
    \\
    & the family & 2021 & 2022 & 2021 & 2022 & 2021
    \\
    \hline
    \# children & -- & 635 & 550 & 1589 & 1372 & 915
    \\
    \hline
    \# daycares & -- & 33 & 33 & 72 & 72 & 54
    \\
    \hline
    \multirow{3}{*}{\# families}
    & $1$ & 542 & 462 & 1331 & 1161 & 777
    \\
    & $2$ & 42 & 44 & 120 & 101 & 66
    \\
    & $3$ & 3 & 0 & 6 & 3 & 2
    \\
    \hline
    \# families & 2 & 3 & 8 & 14 & 25 & 9
    \\
    with twins & 3 & 3 & 0 & 4 & 3 & 1
    \\
    \hline
    \# children & 1 & 41 & 24 & 92 & 66 & 85
    \\
    who prefer & 2 & 20 & 16 & 41 & 27 & 10
    \\
    to transfer & 3 & 0 & 0 & 2 & 2 & 0
    \\
    \hline
\end{tabular}
}
\end{center}
\caption{Numbers of children, daycares, and families in five data sets. The second column corresponds to families with 1,2 or 3 children. For instance, the bottom row calculates the number of children who prefer to transfer from families with 1, 2, or 3 children.}
\label{tab:data}
\end{table}

\begin{table*}[tb]
\centering
\resizebox{1.9\columnwidth}{!}{
\begin{tabular}{ccccccccccc}
    \hline
     & \multicolumn{2}{c}{Tama-21} & \multicolumn{2}{c}{Tama-22} & \multicolumn{2}{c}{Shibuya-21} & \multicolumn{2}{c}{Shibuya-22}
     & \multicolumn{2}{c}{Moriguchi-21}
    \\
        & status quo & IP & status quo & IP & status quo & IP & status quo & IP & status quo & IP
    \\
    \hline
    \small{\# matched children} & 558 & 560 & 464 & 464 & 1307 & 1307 & 1087 & 1087 & 669 & 680
    \\
    \hline
    \small{\# blocking coalition} & 42 & 0 & 2 & 0 & 0 & 0 & 0 & 0 & 116 & 0
    \\
     \hline
    \small{family rationality} & \cmark & \cmark & \cmark & \cmark & \cmark & \cmark & \cmark & \cmark & \xmark & \cmark
    \\   
    \hline
    \small{family non-wastefulness} & \xmark & \cmark & \xmark & \cmark & \cmark & \cmark & \cmark & \cmark & \xmark & \cmark
    \\
     \hline
\end{tabular}
}
\caption{Comparison of outcomes for five data sets}
\label{table:outcome}
\end{table*}

\section{Future Work}

We elaborate on several research questions that would be beneficial for other theoretical models or applications. 
\subsubsection{Flexible Quotas}

One of the main objectives of Japanese daycare matching market is to reduce (or ideally eliminate) the number of unmatched children. A possible approach is to establish flexible quotas instead of hard targets that comply with two feasibility regulations imposed by the government. 
%
%
As shown in two recent work~\citep{Okum19a,KaKo19a}, flexible quotas could significantly increase the number of matched children.

However, based on the feedback from the municipalities, there are two main reasons that flexible quotas are not employed in the current system. Firstly, daycare centers are operated independently and municipalities do not have the authority to interfere in their management. Secondly, children of certain ages require specialized facilities (e.g. crawl space for babies) and thus superfluous room for higher grades cannot be utilized by children of age $0$ or $1$. 

We are still interested in designing algorithms under flexible quotas, as it may be possible to reallocate unused room space among children of certain ages. For example, suppose age $0$ and $1$ form group $1$, age from $2$ to $5$ form group $2$ and room space is shared between each age group.

\subsubsection{Fairness}
Although experimental results on several real-life data sets show existence of stable outcomes in practice, there is no theoretical guarantee that there always exists a stable outcome in general. Thus it is worth studying how to properly define fairness concepts that are compatible with family rationality and family non-wastefulness (or possibly weaker concepts of non-wastefulness). 

Here are several factors that need to be considered when defining fairness: i) should a child or a family has envy towards a child, a family or a set of families? ii) which priority ordering should we consider when
giving precedence to some children or families (daycares' priorities, a common ordering over children, a common ordering over families, or combinations of these priority orderings)? iii) should we confine envy to the set of children who have the same grade only or should we allow envy across grades? iv) when some agent $a$ (a child or a family) has envy towards another agent $b$, should the envy be deemed as justified only if there exists a feasible outcome after agent $a$ replaces agent $b$? 
v) can we compute a fair outcome efficiently, e.g., in polynomial-time? 

\subsubsection{Indifferences} 
Families are requested to submit strict preferences over tuples of daycare centers. Here are two aspects in which indifferences may be superior. Firstly, it is natural and reasonable to assume that a family is indifferent between some options. Secondly, allowing indifferences can increase the number of matched children in theory. The extreme case is that each family submits a binary preference ordering over tuples of daycares (i.e. only acceptable tuples). 

The common ordering over children and daycare priorities over children are derived from priority scores and ties occur often in the data sets. As priorities are important in determining which children take precedence over others, it is unclear how different choices of tie-breaking for priority orderings affect the outcome. 

\subsubsection{Other Properties} 
Pareto optimality is a stronger concept than non-wastefulness, which requires that there is no other outcome in which all agents are weakly better off and at least one agent is strictly better off. As stability and Pareto optimality are generally incompatible for two-sided matching problems~\citep{Roth84a}, we have to give up one of them when designing algorithms. Although it is not the main concern of this work, it is interesting to consider i) whether there exists a Pareto optimal outcome that also satisfies some other properties; ii) whether it promotes the welfare of children.

Strategy-proofness is an important property in mechanism design and has been extensively studied in matching markets. Although it may be impossible to design algorithms that achieve strategy-proofness and other desirable properties for the daycare matching problem, we want to figure out to what extent families are willing to manipulate their strategies (i.e. misreporting their true preferences) in practice.

\section{Acknowledgement}
This work is partially supported by JSPS KAKENHI Grant Number JP20H00609 and JP21H04979.

\bibliography{reference_220812}


\appendix

\section{More on Experiments}

We next give a detailed description of five data sets and summarize the imbalance of demand and supply by age in Table~\ref{table:demand}.

\subsubsection{Tama 2021}

Tama's data set for the year $2021$ involves $635$ children, $587$ families and $33$ daycares. 
There are $42$ families with two children and $3$ out of them have a pair of twins.
Only $3$ families have three children and all of them contain a pair of twins. There are $61$ children are initially enrolled and $20$ out of them have siblings. The age distribution of the children is $28.50\%$, $40.47\%$, $15.43\%$, $11.81\%$, $2.68\%$ and $1.10\%$ corresponding to ages from $0$ to $5$.
For families with an only child, the average length of family preferences is $3.3$ and the maximum length is $15$. For families with siblings, the average length of family preferences is $38.37$ and the maximum length is $1088$.

\begin{table}[tb]
\centering
\resizebox{\columnwidth}{!}{
\begin{tabular}{ccccccccc}
    \hline
    & age & $0$ & $1$ & $2$ & $3$ & $4$ & $5$
     \\
    \hline
    \multirow{2}{*}{Tama-21}
    & \# applicants & $181$ & $257$ & $98$ & $75$ & $17$ & $7$
     \\
    & \# capacity   & $241$ & $222$ & $123$ & $106$ & $57$ & $68$
     \\
    \hline
    \multirow{2}{*}{Tama-22}
    & \# applicants & $181$ & $219$ & $91$ & $43$ & $8$ & $8$
     \\
    & \# capacity   & $230$ & $212$ & $88$ & $83$ & $41$ & $58$
     \\
    \hline
    \multirow{2}{*}{Shibuya-21}
    & \# applicants & $569$ & $656$ & $171$ & $136$ & $37$ & $20$
     \\
    & \# capacity   & $509$ & $613$ & $239$ & $265$ & $268$ & $275$
     \\
    \hline
    \multirow{2}{*}{Shibuya-22}
    & \# applicants & $540$ & $582$ & $134$ & $67$ & $33$ & $16$
     \\
    & \# capacity   & $497$ & $586$ & $186$ & $233$ & $255$ & $306$
     \\
    \hline
        \multirow{2}{*}{\small{Moriguchi-21}}
    & \# applicants & $257$ & $354$ & $184$ & $89$ & $17$ & $14$
     \\
    & \# capacity & $369$ & $294$ & $156$ & $66$ & $38$ & $38$
     \\
    \hline
\end{tabular}
}
\caption{Demand and supply by age}
\label{table:demand}
\end{table}
\subsubsection{Tama 2022}

Tama's data set for the year $2022$ involves $550$ children, $506$ families and $33$ daycares. 
No family is endowed with three children and $462$ families has an only child. The remaining $44$ families have two children and contain $8$ pairs of twins. Only $40$ children are initially enrolled and $16$ out of them have siblings. The age distribution of the children is $32.91\%$, $39.82\%$, $16.55\%$, $7.82\%$, $1.45\%$ and $1.45\%$ corresponding to ages from $0$ to $5$.
For families with an only child, the average length of family preferences is $3.0$ and the maximum length is $8$. For families with siblings, the average length of family preferences is $8.4$ and the maximum length is $64$.

\subsubsection{Shibuya 2021}

Shibuya's dataset of the year $2021$ involves $1589$ children, $1457$ families and $72$ daycares. 
The number of families with two children is $120$ and $14$ of them have twins. There are $6$ families with three children and $4$ of them contain a pair of twins. The number of children with initial enrollment is $135$ and $43$ of them have siblings. The age distribution of the children is $35.81\%$, $41.28\%$, $10.76\%$, $8.56\%$, $2.33\%$, $1.26\%$ corresponding to ages from $0$ to $5$.
For families with an only child, the average length of family preferences is $4.45$ and the maximum length is $11$. For families with siblings, the average length of family preferences is $14.95$ and the maximum length is $120$.

\subsubsection{Shibuya 2022}

Shibuya's dataset of the year $2022$ involves $1372$ children, $1265$ families and $72$ daycares. 
The number of families with two children is $101$ and $25$ of them are twins. There are $3$ families with three children and all of them contain a pair of twins. The number of children with initial enrollment is $95$ and $29$ out of them have siblings. The age distribution of the children is $39.36\%$, $42.42\%$, $9.77\%$, $4.88\%$, $2.40\%$, $1.17\%$ corresponding to ages from $0$ to $5$.
For families with an only child, the average length of family preferences is $3.78$ and the maximum length is $10$. For families with siblings, the average length of family preferences is $6.58$ and the maximum length is $64$.

\subsubsection{Moriguchi 2021}

Moriguchi's dataset of the year $2021$ involves $915$ children, $845$ families and $54$ daycares. 
The number of families with two children is $66$ and $9$ of them are twins. There are $2$ families with three children and one of them contains a pair of twins. The number of children with initial enrollment is $95$ and $10$ out of them have siblings. The age distribution of the children is $28.09\%$, $38.69\%$, $20.11\%$, $9.72\%$, $1.86\%$, $1.53\%$ corresponding to ages from $0$ to $5$.
For families with an only child, the average length of family preferences is $2.56$ and the maximum length is $5$. For families with siblings, the average length of family preferences is $5.20$ and the maximum length is $24$.

\subsection{Other Experimental Results}

We next present more experimental results where only some of the IP constraints are satisfied. 
We refer to these new algorithms
as IP-X where the difference is whether we allow blocking coalition for certain families. The fundamental constraints include feasibility, family rationality, and family non-wastefulness, which are captured by the IP-0 algorithm. IP-1 indicates that there is no blocking coalition for families with one child.  IP-2 indicates that there is no blocking coalition for families with one and two children. IP-3 indicates that there is no blocking coalition for families with one, two and three children.
Note that the number of blocking coalitions may be different, as we do not consider it in the objective function.

\begin{table}[tb]
\centering
\resizebox{\columnwidth}{!}{
\begin{tabular}{cccccc}
    \hline
     &  status quo & IP-3 & IP-2 & IP-1 & IP-0
     \\
     \hline
    family rationality & \cmark & \cmark & \cmark & \cmark & \cmark
    \\   
    \hline
    family non-wastefulness & \xmark & \cmark & \cmark & \cmark & \cmark
    \\
     \hline
     \# matched children & 558 & 560 & 560 & 569 & 595
     \\
     \hline
     \# blocking coalition & 42 & 0 & 0 & 31 & 373
     \\
     \hline
     %
    \# running time (s) & --  & 1.86 & 1.87 & 2.78 & 49
     \\
    \hline
\end{tabular}}
\caption{Outcomes for Tama 2021}
\label{table:Tama21}

\resizebox{\columnwidth}{!}{
\begin{tabular}{cccccc}
    \hline
     &  status quo & IP-3 & IP-2 & IP-1 & IP-0
     \\
     \hline
    family rationality & \cmark & \cmark & \cmark & \cmark & \cmark
    \\   
    \hline
    family non-wastefulness & \xmark & \cmark & \cmark & \cmark & \cmark
    \\
     \hline
     \# matched children & 464 & 464 & 464 & 469 & 509
     \\
     \hline
     \# blocking coalition & 2 & 0 & 0 & 30 & 336
     \\
    \hline
   \# running time (s) & --  & 0.16 & 0.15 & 0.62 & 1.5
     \\
    \hline
\end{tabular}
}
\caption{Outcomes for Tama 2022}
\label{table:Tama22}

\resizebox{\columnwidth}{!}{
\begin{tabular}{cccccc}
    \hline
     &  status quo & IP-3 & IP-2 & IP-1 & IP-0
     \\
     \hline
    family rationality & \cmark & \cmark & \cmark & \cmark & \cmark
    \\   
    \hline
    family non-wastefulness & \cmark & \cmark & \cmark & \cmark & \cmark
    \\
     \hline
     \# matched children & 1307 & 1307 & 1312 & 1339 & 1464
     \\
     \hline
     \# blocking coalition & 0 & 0 & 7 & 311 & 2481
     \\
    \hline
     \# running time (s) & -- & 1.85 & 3.43 & 2526 & 480
     \\
    \hline
\end{tabular}
}
\caption{Outcomes for Shibuya 2021}
\label{table:Shibuya21}

\resizebox{\columnwidth}{!}{
\begin{tabular}{cccccc}
    \hline
     &  status quo & IP-3 & IP-2 & IP-1 & IP-0
     \\
     \hline
    family rationality & \cmark & \cmark & \cmark & \cmark & \cmark
    \\   
    \hline
    family non-wastefulness & \cmark & \cmark & \cmark & \cmark & \cmark
    \\
     \hline
     \# matched children & 1087 & 1087 & 1087 & 1111 & 1206
     \\
     \hline
     \# blocking coalition & 0 & 0 & 0 & 86 & 1530
       \\
    \hline
    \# running time (s) & -- & 0.57 & 0.58 & 58.7 & 23.8
     \\
    \hline
\end{tabular}
}
\caption{Outcomes for Shibuya 2022}
\label{table:Shibuya22}

\centering
\resizebox{\columnwidth}{!}{
\begin{tabular}{cccccc}
    \hline
     &  status quo & IP-3 & IP-2 & IP-1 & IP-0
     \\
     \hline
    family rationality & \xmark & \cmark & \cmark & \cmark & \cmark
    \\   
    \hline
    family non-wastefulness & \xmark & \cmark & \cmark & \cmark & \cmark
    \\
     \hline
      \# matched children & 669 & 680 & 680 & 693 & 746
     \\
     \hline
     \# blocking coalition & 116 & 0 & 0 & 46 & 599
      \\
    \hline
     \# running time (s) & -- & 0.19 & 0.18 & 0.42 & 2.64
     \\
    \hline
\end{tabular}
}
\caption{Outcomes for Moriguchi 2021}
\label{table:Moriguchi21}
\end{table}

\end{document}